\documentclass{article}
\usepackage[margin=2.5cm]{geometry}
\usepackage{color}
\usepackage{amsfonts,amsmath,amsthm}

\usepackage{hyperref}
\usepackage{latexsym}
\usepackage{algorithmic}
\usepackage{algorithm}

\begin{document}
\newtheorem{thm}{Theorem}
\newtheorem{lem}[thm]{Lemma}
\newtheorem{cor}[thm]{Corollary}
\newtheorem{prop}[thm]{Proposition}

\newtheorem{defn}{Definition}

\newtheorem{rem}{Remark}
\newtheorem{ex}{Example}

\title{Design and Analysis of an Estimation of Distribution Approximation Algorithm for  Single Machine Scheduling  in Uncertain Environments}

\author{Boris Mitavskiy\thanks{Corresponding author.} and Jun He  \\Department of Computer Science\\ Aberystwyth University, Aberystwyth, SY23 3DB, U.K.  }

\maketitle

\begin{abstract}
In the current work we introduce a novel  estimation of distribution algorithm  to tackle a hard combinatorial optimization problem, namely the single-machine scheduling problem, with uncertain delivery times. The majority of the existing research coping with optimization problems in uncertain environment aims at finding a single sufficiently robust solution so that random noise and unpredictable circumstances  would have the least possible detrimental effect on the quality of the solution. The measures of robustness are usually based on various kinds of empirically designed averaging techniques. In contrast to the previous work, our algorithm aims at finding a collection of robust schedules that allow for a more informative decision making. The notion of robustness is measured quantitatively in terms of the classical mathematical notion of a norm on a vector space. We provide a theoretical insight into the relationship between the properties of the probability distribution over the uncertain delivery times and the robustness quality of the schedules produced by the algorithm after a polynomial runtime in terms of approximation ratios.
\end{abstract}

\section{Introduction}\label{IntroductionSect}
Scheduling problems naturally arise in a number of disciplines, for instance, in computer science and in operations research. Most scheduling problems are NP-hard and, therefore, no efficient algorithms exist to find an optimal solution for them unless P=NP (see \cite{hochbaum1996approximation}). In order to obtain some satisfactory solutions to such problems, evolutionary algorithms (EAs) are widely applied. A popular NP-hard scheduling problem will be discussed in the paper, which is the so-called single machine scheduling problem where $n$ jobs, each having their own release times, processing times and delivery times have to be processed on a single machine without preemption.

Current research in EAs for scheduling in uncertain environments  shares the following drawback: the EAs for solving scheduling problems are designed to seek a single solution to scheduling in a random environment based on some kind of a notion of the best suited robust schedule (see surveys~\cite{jin2005evolutionary,bianchi2009survey}). For example, in \cite{goren2008robustness} surrogate notion of \emph{robustness} of a solution to a single machine scheduling problem in uncertain environments has been introduced. According to \cite{goren2008robustness}, ``a schedule whose performance does not significantly degrade in the face of disruption is called robust". Mathematically, robustness of a given schedule constructed on a fixed instance of a problem is compared in terms of expected values of the deviations from the performance of the schedule due to random changes in the environment.

We believe a single solution is not sufficiently informative in the uncertain environments in the sense that it does not reflect the scenario in real world. A single optimal solution only corresponds to a special situation, such as the most probable situation or, even worse, just some kind of an average among all possible situations.

Unlike the previous investigators, we aim to explore an entire set of solutions which correspond to potentially different situations in a random environment. The primary objective of our EA is to produce a population of ordered tuples of the form $(instance, \; schedule, \; other \, information)$ such that

\begin{enumerate}
\item The instances contained within the individuals in the final generation (i.e. after finishing the run of the algorithm) are, in a certain sense to be made precise later, the most suitable instances on which ``robust" schedules should be constructed.

\item  After the EA finishes it's run, with high probability, a randomly selected instance will be ``sufficiently close" to at least one of the instances contained within the individuals in the final population so that the robust schedule constructed on that instance would then be also be sufficiently good for the corresponding randomly sampled instance.

\end{enumerate}

As mentioned in \cite{sotskov2010scheduling}, in uncertain environments there are two phases: one is a training phase when the machine simulates the trials with the aim of making decisions later, while the next phase is the actual decision making in a practical situation. The time complexity during the training phase is not required to be nearly as fast as the time complexity while making a decision on a ``sufficiently good" solution in a practical situation. In the current article, we are going to pin down the crucial properties of the unknown probability distributions for which our EA performs well in polynomial runtime. Furthermore, we explain theoretically how to tune the parameters of our EA given heuristically estimated or expert knowledge about the properties of the unknown probability distributions. An EA of the type we propose is known as an EDA (estimation of distribution algorithm). The main difference between EDAs and most conventional evolutionary algorithms is that evolutionary algorithms generate new candidate solutions using an implicit distribution defined by one or more variation operators, whereas EDAs use an explicit probability distribution encoded by a Bayesian network, a multivariate normal distribution, or another model class. In the current article, the probability distribution guiding the search is the one on the uncertain parameters of the problem: the delivery times in a single machine scheduling problem. In recent years EDAs have gained significant popularity: see, for instance, \cite{santana2009research} and \cite{lozano2006towards}.

Once the EDA has finished running, if we are given polynomially many parallel processors, one can check in linear runtime which individual in the final population contains a ``sufficiently good" schedule. A precise meaning of ``sufficiently good" is provided in terms of the classical notion of approximation ratios (see \cite{auger2011theory}).

The remainder of the paper is organized as follows. In Section~\ref{secReivew}   we provide a substantial literature review related to EA-based applications for scheduling problems in deterministic as well as in uncertain environments.
In Section~\ref{basicPropSect} we introduce the classical single machine scheduling problem. In Section~\ref{EuclidEmbedSect} single machine scheduling problem with uncertain delivery times is introduced and an important theoretical link between the ``maximum of absolute values of coordinates"-norm and the notion of robustness is established. In Section~\ref{algorithmSect}, we describe the EDA in detail. Section~\ref{ProbDistributionAnalysisSect} is devoted to the theoretical analysis of the novel EDA and Section~\ref{concludeSect} summarizes the results and suggests future work.
\section{Literature Review}
\label{secReivew}
\subsection{Literature Review of EA-based Applications to Scheduling problems in Deterministic Environments}\label{LiteratureReviewDetermIntro}
EAs have been widely applied to tackle the single machine scheduling problem for over two decades. For example, an approach for solving the single machine scheduling problem with an objective to minimize the flow time variance based on genetic algorithms has been proposed in \cite{gupta1993minimizing} as far back in time as in the early 1990s. In order to improve the efficiency, most EAs employ different search strategies. For example, a memetic algorithm for the total tardiness single machine scheduling problem with due dates has been proposed in \cite{francca2001memetic}. Several recombination operators, local improvement procedure and several neighborhood reduction schemes have been used. Genetic algorithms  with four types of crossover operators and three types of mutation operators plus local search have been proposed in \cite{sevaux2003genetic}. A discrete differential evolution algorithm has been exploited to tackle the single machine total weighted tardiness problem in \cite{tasgetiren2009discrete}. The algorithm combined constructive heuristics such as the well-known NEH heuristic and greedy randomized adaptive search procedure as well as some priority rules.
Single machine scheduling problem with the objective of minimizing the maximal lateness has been considered in \cite{sels2012hybrid}. A hybrid genetic algorithm with the combination of different local search neighborhood structures has been exploited.

A few articles have been devoted to comparing different types of EAs. They claimed that some EAs are superior to others. Approaches based on evolutionary strategies, simulated annealing and threshold acceptance for solving the problem of scheduling a number of jobs on a single machine against a restrictive common due date have been compared in \cite{feldmann2003single} It has been concluded that a new variant of threshold acceptance is superior to the other approaches. A genetic algorithm, a simulated annealing approach, local search methodology, branch-and-bound algorithm, and ant colony optimization algorithms, for solving a single machine scheduling problem have been compared in \cite{gagne2002comparing}. Their results indicated that ant colony optimization is competitive and has a certain advantage on larger problems. About the only article proposing a rigorously established polynomial-time approximation scheme for a single machine scheduling problem based on an evolutionary algorithm appears in \cite{mitavskiy2012polynomial}.

All of the literature discussed above assumes that all of the parameters (for example, setup times, processing times, releases dates, due dates and other parameters) are deterministic. Nonetheless, in the real world situations, a number of scheduling tasks involve some degree of uncertainty. For example, the delivery time of groceries is affected by uncertain traffic; the arrival time of an airplane is delayed by uncertain weather etc. A more reasonable model would be to consider some of the parameters as the random variables on the probability space of all possible practical scenarios. The relevant literature review appears in the next subsection.
\subsection{Literature Review of EA-based Applications to Scheduling problems in Uncertain Environments}\label{LiteratureReviewUncertainIntro}
In recent years, EAs have been emerging as a successful  tool to cope with optimization problems in uncertain environments. A comprehensive review on the topic appears in \cite{jin2005evolutionary,bianchi2009survey}. Currently there are two main approaches to handle the uncertainty in EAs: one is based on probability theory and the other one on fuzzy set theory.

The more popular approach to deal with uncertainty is based on probability theory. The goal is to optimize the expectation of an objective function with certain random inputs. For example, a genetic algorithm for solving stochastic job-shop scheduling problems has been proposed in \cite{yoshitomi2002genetic}.
\cite{gu2009novel} constructed a parallel quantum genetic algorithm for the stochastic job shop scheduling problem with the objective of minimizing the expected value of makespan, where the processing times are subjected to independent normal distributions.
\cite{wang2009robust}  solved the logistics center location and allocation problem under uncertain environment using enumeration method, and genetic algorithm. Both, the deterministic optimization model and a two-stage stochastic optimization model have been considered. A competitive co-evolutionary quantum genetic algorithm for a stochastic job shop scheduling problem with the objective to minimize the expected value of makespan has been proposed in \cite{gu2010novel}.

Multi-objective EAs are also used in the study of scheduling in uncertain environment. For example, \cite{lei2011simplified} investigated multi-objective stochastic job shop scheduling problem.
\cite{ebrahimi2013hybrid} presented the non-dominated sorting genetic algorithm  for  hybrid flow shop  scheduling with  uncertain due date  in a normal distribution.

Fuzzy set theory provides an alternative way of handling the uncertainty. For example,
\cite{peng2004parallel} presented three types of fuzzy scheduling models for parallel machine scheduling problems with fuzzy processing times.
\cite{fayad2005fuzzy} proposed a multi-objective genetic algorithm to deal with a real-world fuzzy job shop scheduling problem. Fuzzy sets are used to model uncertain due dates and processing times of jobs.
\section{Single Machine Scheduling Problem in Deterministic Environments}\label{basicPropSect}
\subsection{A Preliminary Description of the Single-Machine Scheduling Problem}\label{IntroSubsect}
The presentation in this subsection closely resembles that in~\cite{hochbaum1996approximation}. The single machine scheduling problem is well-known to be NP-hard (see \cite{hochbaum1996approximation}). In the classical form it can be described as follows. Suppose we have a sequence of $n$ jobs $J = \{J_i\}_{i=1}^n$ where each job $J_i$ must be processed without interruption for a time $p_i > 0$ on the same machine. A job $J_i$ is released at time $r_i \geq 0$ associated with it and it becomes available for processing only at the time $r_i$ and any time after as long as the machine is not occupied at the time being. As soon as a job $J_i$ has finished processing it is sent for delivery immediately. Notice that the jobs are allowed to be delivered in parallel without any restriction, yet they can be processed only sequentially. A specific job $J_i$ has its own \emph{delivery time} $q_i$. Here we assume that there is no restriction on the total number of jobs being delivered simultaneously.\footnote{It has been explicitly shown in~\cite{hochbaum1996approximation} that the setting above is equivalent to the model with due dates in place of the delivery times via a simple linear change of variables, yet the model with delivery times is a lot better suitable for the algorithm design and analysis.} Informally speaking, our objective is to find a reordering of the sequence $\{J_i\}_{i=1}^n$ of jobs which minimizes the minimal time when all of the jobs have just been delivered, referred to as the maximal lateness of the schedule. In order to understand the problem better and to design efficient algorithms, it is necessary to describe the notions of ``reordering", ``maximal lateness of a schedule", the ``minimal maximal lateness" of an instance of the single machine scheduling problem, etc. in more detail. This will be the subject of the next subsection.
\subsection{A Detailed Description of the Main Objectives of the Single-Machine Scheduling Problem.}\label{formalDescribeSubsect}
A formal mathematical formulation of reordering can be given in terms of a permutation of the indices, $\pi$ on $\{1, \, 2, \ldots, n\}$. Every such permutation defines a new ordering in which the jobs are to be processed,  $J_{\pi} = \{J_{\pi(i)}\}_{i=1}^n$ (i.e. a new schedule with respect to the originally given schedule determined by the permutation $\pi$). As soon as the job $\pi(1)$ is released at the time $r_{\pi(1)}$ it starts getting processed and it takes time $p_{\pi(1)}$ to process the job $J_{\pi(1)}$. We denote by $s_{\pi(1)} = r_{\pi(1)}$ the \emph{starting time} of the job $J_{\pi(1)}$.\footnote{Notice that it makes no sense to wait as long as both, the machine and the job to be scheduled next, are available} Since the \emph{processing time} of the job is $J_{\pi(1)}$ $p_{\pi(1)}$, the machine is not available until it has finished processing the job $J_{\pi(1)}$ at the time $s_{\pi(1)} + p_{\pi(1)}$. As soon as the job $J_{\pi(1)}$ has been processed, it is immediately sent for delivery and it takes time $q_{p(1)}$ to deliver the job $J_{\pi(1)}$. Thus the job $J_{\pi(1)}$ is delivered at the time $s_{\pi(1)} + p_{\pi(1)} + q_{\pi(1)}$. At the same time as the job $J_{\pi(1)}$ has been sent for delivery, i.e. at the time $s_{\pi(1)} + p_{\pi(1)}$, the machine becomes available and, if the job $J_{\pi(2)}$ is available at the time $s_{\pi(1)} + p_{\pi(1)}$, (i.e. if $r_{\pi(2)} < s_{\pi(1)} + p_{\pi(1)}$) it starts getting processed. Otherwise, the job $J_{\pi(2)}$ starts getting processed at the time $r_{\pi(2)}$. Thus, the \emph{starting time} of the job $J_{\pi(2)}$ is $s_{\pi(2)} = \max\{s_{\pi(1)} + p_{\pi(1)}, \, r_{\pi(2)}\}$. Continuing in this manner recursively, the starting time $s_{\pi(i+1)}$ of the job $J_{\pi(i+1)}$ is $\max\{s_{\pi(i)} + p_{\pi(i)}, \, r_{\pi(i+1)}\}$ while the time when the job $J_{\pi(i)}$ has been delivered is $s_{\pi(i)} + p_{\pi(i)} + q_{\pi(i)}$. In case when $r_{\pi(i+1)} > s_{\pi(i)} + p_{\pi(i)}$ we say that the \emph{machine is idle before the job} $J_{\pi(i+1)}$ \emph{has started processing}. The \emph{maximal lateness} of the schedule $J_{\pi}$ is defined as $J_{\pi}^{\max} = \max_{1 \leq i \leq n}\{s_{\pi(i)} + p_{\pi(i)} + q_{\pi(i)}\}$, i.e. the time when all the jobs have just been delivered. In summary we are given a sequence of $n$ ordered triplets $\{J_i\}_{i=1}^n$ with each $J_i = (r_i, \, p_i, \, q_i)$ that stand for \emph{release time}, \emph{processing time} and \emph{delivery time} of the job $J_i$ respectively. The \emph{minimal maximal lateness} of the instance $J$ of the single-machine scheduling problem, denoted by $J^*$, is defined as $$J^* = \min\{J_{\pi}^{\max} \, | \, \pi \text{ is a permutation on } \{1, \, 2, \ldots n\}\}.$$
The goal of the single-machine scheduling problem can now be reformulated rigorously as follows. We aim to find a permutation $\sigma$ on $\{1, \, 2, \ldots, n\}$ such that  $J_{\sigma}^{\max} = J^*$.

Since the problem is NP-hard, one hopes to find a ``satisfactory" solution in polynomial runtime. The notion of a ``satisfactory" solution is often given in terms of an \emph{approximation ratio up to a given factor} $r \geq 1$ (see, for instance, \cite{hochbaum1996approximation}). This means that we aim to find a schedule $\pi$ such that $\frac{J_{\pi}^{\max}}{J^*} \leq r$. Of course, it is desirable to make $r$ as close to $1$ as practically feasible. Despite the single-machine scheduling problem is NP-hard, it is, as described in \cite{hochbaum1996approximation}, ``one of the easiest NP-hard problems" in the sense that there exists a \emph{polynomial-time approximation scheme} for it. This means that there is an algorithm, call it $\mathcal{A}$, such that $\forall \epsilon > 0$ $\exists$ a sufficiently large $B(\epsilon) \in \mathbb{N}$ with the property that after at most $n^{B(\epsilon)}$ time steps the algorithm $\mathcal{A}$ finds a permutation-schedule $\sigma$ on $\{1, \, 2, \ldots, n\}$ with $\frac{J_{\sigma}^{\max}}{J^{*}} \leq 1+\epsilon$.\footnote{Of course, $B(\epsilon) \rightarrow \infty$ as $\epsilon \rightarrow 0$ (or, equivalently, as the approximation ratio $r = 1+\epsilon \rightarrow 1$).} Several polynomial-time approximation schemes for the single-machine scheduling problem can be found in the literature (see \cite{hochbaum1996approximation} for details). A polynomial time approximation scheme for the single-machine scheduling problem based on an evolutionary algorithm can be found in \cite{mitavskiy2012polynomial}.
\section{Single Machine Scheduling Problem in Uncertain Environments}\label{EuclidEmbedSect}
In the current section, we describe the single machine scheduling problem in uncertain environments and introduce a novel methodology to cope with the uncertainty.
\subsection{Mathematical Framework and Notation}\label{frameworkSubsect}
In this article we will assume that the delivery times are random variables drawn from an unknown probability distribution while the processing times and the release times are fixed throughout the problem. Thus, every instance of the single machine scheduling problem of size $n$, can be written as $\theta = (J(\theta)_1, \, J(\theta)_2, \ldots, J(\theta)_n)$ with $J(\theta)_i = (r_i, \, p_i, \, q(\theta)_i)$ where $r_i$ and $p_i$ are fixed over the entire space of instances of the single-machine scheduling problem of size $n$, while $q(\theta)_i$s are drawn from an unknown probability distribution. Since all instances of size $n$ share common release times and processing times, they are fully determined by the $n$-dimensional \emph{delivery times vectors} or, also referred to as \emph{vectors of delivery times}, $$\vec{q}(\theta) = (q(\theta)_1, \, q(\theta)_2, \ldots, q(\theta)_n) \in [0, \, \infty)^n \subseteq \mathbb{R}^n.$$ The following definition will be convenient throughout the remainder of the article.

\begin{defn}\label{vectorInducedInstanceDefn}
Suppose we are given a vector $\vec{v} = (v_1, \, v_2, \ldots v_n) \in [0, \, \infty)^n$. We will say that the vector $\vec{v}$ induces the instance of the single machine scheduling problem $\theta$ or, alternatively, that the instance of the single machine scheduling problem $\theta$ is induced by the vector $\vec{v}$ if delivery times vector of the instance $\theta$, namely $\vec{q}(\theta) = \vec{v}$. In particular, every instance of the single machine scheduling problem of size $n$ is induced by its vector of delivery times.
\end{defn}

Recall from   Section~\ref{IntroductionSect} that robustness of a schedule plays a crucial role in uncertain environments. Consider now the \\
$\| \cdot \|_{\infty}$ on $\mathbb{R}^n$ that's traditionally defined as the maximal absolute value over all coordinates. More precisely, for a vector $\vec{x} = (x_1, \, x_2, \ldots, x_n) \in \mathbb{R}^n$ $\|\vec{x}\|_{\infty} = \max_{1 \leq i \leq n}|x_i|$. In the next subsection we are going to present two very simple, yet rather powerful facts that relate robustness of schedules to the $\| \cdot \|_{\infty}$.

\subsection{Exhibiting a Strong Link Between Robustness of Schedules and the Infinity Norm}\label{robustLinkSubsect}
In the current subsection we will demonstrate that the minimal maximal lateness, treated as a non-negative valued function from $[0, \, \infty)^n \rightarrow \mathbb{R}$ is Lipschitz-continuous with respect to the $\|\cdot \|_{\infty}$ norm on $\mathbb{R}_n$ and the Lipschitz constant is $1$. Furthermore, the optimal solutions (i.e. permutation-schedules) for a given instance of the single-machine scheduling problem, call it $\theta$, exceed the minimal maximal lateness of any other instance $\gamma$ by at most $\|\vec{q}(\theta) - \vec{q}(\gamma)\|_{\infty}$. Before proceeding any further, we point out the following simple observation that will be useful throughout the current section.

\begin{rem}\label{identicalStartingTimesRem}
Given a permutation-schedule $\pi$ and instances $\theta_1$ and $\theta_2$ of the single-machine scheduling problem having identical release times and processing times for all the jobs, denote the corresponding jobs as $J(\theta_1)_{\pi(i)} = (r_{\pi(i)}, \, p_{\pi(i)}, q(\theta_1)_{\pi(i)})$ and $J(\theta_2)_{\pi(i)} = (r_{\pi(i)}, \, p_{\pi(i)}, q(\theta_2)_{\pi(i)})$. Since the schedule $\pi$ has already been selected, the starting times of all the jobs in the schedule depend only on the release times and the processing times but not on the delivery times according to the definition of the single-machine scheduling problem. It then follows that the starting times $s_{\pi(i)}$ are the same for both instances, $\theta_1$ and $\theta_2$.
\end{rem}

Remark~\ref{identicalStartingTimesRem} motivates the crucial observation mentioned above.

\begin{prop}\label{LipshitzNormProp}
Suppose we are given two instances $\theta_1$ and $\theta_2$ of the single-machine scheduling problem of size $n$ that have identical release dates and processing times for all the jobs, but the delivery times may be distinct. Then it follows that $\forall$ permutation-schedule $\pi$ we have $|J(\theta_1)_{\pi}^{\max} - J(\theta_2)_{\pi}^{\max}| \leq \|\vec{q}(\theta_1) - \vec{q}(\theta_2)\|_{\infty}$. Consequently, if $\sigma$ is a permutation-schedule achieving the maximal lateness of the instance $\theta_1$ then $|J(\theta_1)^{*} - J(\theta_2)_{\sigma}^{\max}| \leq \|\vec{q}(\theta_1) - \vec{q}(\theta_2)\|_{\infty}$ and, symmetrically, if $\rho$ is a permutation-schedule achieving the maximal lateness of the instance $\theta_2$ then $|J(\theta_2)^{*} - J(\theta_1)_{\rho}^{\max}| \leq \|\vec{q}(\theta_1) - \vec{q}(\theta_2)\|_{\infty}$. In particular, $|J(\theta_1)^{*} - J(\theta_2)^{*}| \leq \|\vec{q}(\theta_1) - \vec{q}(\theta_2)\|_{\infty}$.
\end{prop}
\begin{proof}
Indeed, let $\pi$ be a permutation schedule and consider the job $J_{\pi(i)} = (r_{\pi(i)}, \, p_{\pi(i)}, q(\theta_1)_{\pi(i)})$ which achieves the maximal lateness of the permutation-schedule $\pi$ on the instance $\theta_1$, namely, $J(\theta_1)_{\pi}^{\max}$. According to Remark~\ref{identicalStartingTimesRem}, the starting times of all the jobs, $s_{\pi(i)}$ are also identical. Then, by definition of the maximal lateness of a schedule, it follows that $$J(\theta_1)_{\pi}^{\max} = s_{\pi(i)} + p_{\pi(i)} + q(\theta_1)_{\pi(i)} \overset{\text{by definition of }\| \cdot \|_{\infty}}{\leq} $$$$ s_{\pi(i)} + p_{\pi(i)} + \left(q(\theta_2)_{\pi(i)} + \|\vec{q}(\theta_1) - \vec{q}(\theta_2)\|_{\infty}\right) = $$$$=\left(s_{\pi(i)} + p_{\pi(i)} + q(\theta_2)_{\pi(i)}\right) + \|\vec{q}(\theta_1) - \vec{q}(\theta_2)\|_{\infty} \leq$$$$\overset{\text{by definition of }J(\theta_2)_{\pi}^{\max}}{\leq} J(\theta_2)_{\pi}^{\max} + \|\vec{q}(\theta_1) - \vec{q}(\theta_2)\|_{\infty}.$$ In summary, we have shown that $\forall$ permutation schedule $\pi$ we have
\begin{equation}\label{inequalityInNormProof}
J(\theta_1)_{\pi}^{\max} \leq J(\theta_2)_{\pi}^{\max} + \|\vec{q}(\theta_1) - \vec{q}(\theta_2)\|_{\infty}.
\end{equation}
Interchanging the roles of $\theta_1$ and $\theta_2$ in (\ref{inequalityInNormProof}) immediately implies
\begin{equation}\label{inequalityInNormProof1}
J(\theta_2)_{\pi}^{\max} \leq J(\theta_1)_{\pi}^{\max} + \|\vec{q}(\theta_1) - \vec{q}(\theta_2)\|_{\infty}.
\end{equation}
A combination of  (\ref{inequalityInNormProof}) and (\ref{inequalityInNormProof1}) is precisely the first desired conclusion that $$|J(\theta_1)_{\pi}^{\max} - J(\theta_2)_{\pi}^{\max}| \leq \|\vec{q}(\theta_1) - \vec{q}(\theta_2)\|_{\infty}.$$ The next two conclusions are apparent within the statement of the theorem. To see the last conclusion that $$|J(\theta_1)^{*} - J(\theta_2)^{*}| \leq \|\vec{q}(\theta_1) - \vec{q}(\theta_2)\|_{\infty},$$ let $\sigma$ be a permutation-schedule achieving the maximal lateness of the instance $\theta_1$ and $\rho$ be a permutation-schedule achieving the maximal lateness of the instance $\theta_2$, as in the conditions of Proposition~\ref{LipshitzNormProp}. Recall that from the previous two conclusions we have $$|J(\theta_1)^{*} - J(\theta_2)_{\sigma}^{\max}| \leq \|\vec{q}(\theta_1) - \vec{q}(\theta_2)\|_{\infty}$$ and $$|J(\theta_2)^{*} - J(\theta_1)_{\rho}^{\max}| \leq \|\vec{q}(\theta_1) - \vec{q}(\theta_2)\|_{\infty}$$ together with the definition of the minimal maximal lateness it follows that $$J(\theta_2)^{*} \leq J(\theta_2)_{\sigma}^{\max} \leq J(\theta_1)^{*} + \|\vec{q}(\theta_1) - \vec{q}(\theta_2)\|_{\infty}$$ and, symmetrically, $$J(\theta_1)^{*} \leq J(\theta_1)_{\rho}^{\max} \leq J(\theta_2)^{*} + \|\vec{q}(\theta_1) - \vec{q}(\theta_2)\|_{\infty}$$ so that the last conclusion follows at once.
\end{proof}
Proposition~\ref{LipshitzNormProp} is a cornerstone behind the design of the EDA in Section~\ref{algorithmSect}. It provides a quantitative measure of robustness of a schedule on an instance of the single-machine scheduling problem in terms of the $\| \cdot \|_{\infty}$. Indeed, if a permutation-schedule $\pi$ has been constructed on an instance $\theta$ of the single machine scheduling problem with an approximation ratio $r \geq 1$ (recall from Subsection~\ref{formalDescribeSubsect} that $J(\theta)_{\pi}^{\max} \leq r \cdot J(\theta)^{*}$), then, for any other instance $\rho$ the delivery times vector of which is at most $\epsilon$-away from $\vec{q}(\theta)$ as measured by the $\| \cdot \|_{\infty}$, the same schedule $\pi$ produces an approximation at least as good as $$J(\rho)_{\pi}^{\max} \leq J(\theta)_{\pi}^{\max} + \epsilon \leq r \cdot J(\theta)^{*} + \epsilon \leq r \cdot (J(\rho)^{*} + \epsilon) + \epsilon.$$

Apparently it is not possible to construct schedules with good approximation ratios for every instance of the single-machine scheduling problem since there are infinitely many of them. In fact, even if we assume that the delivery times vectors are contained within the $n$-dimensional cube of the form $[0, \, M]^n$ with $M \in (0, \, \infty)$, it would require $2^n$ cubes with side lengths $\frac{M}{2}$ to cover the cube $[0, \, M]^n$. It might often happen in practice, nonetheless, that the vectors of delivery times are most likely to arise within only polynomially many relatively small-sized $n$-dimensional cubes contained inside the big cube $[0, \, M]^n$. In view of Proposition~\ref{LipshitzNormProp}, it would then be preferable to discover the centers of these cubes and to construct near-optimal schedules for the instances of the single-machine scheduling problem induced by these centers in the sense of Dection~\ref{vectorInducedInstanceDefn}. A naturally arising question is then how do we search for these centers efficiently? From the probability-theoretic point of view, the ``small-sized" cubes where the vectors of delivery times most likely occur, correspond to the events happening with relatively high probability with respect to the unknown probability distribution on the vectors of delivery times, while the centers of such cubes are the expectations of the conditional probability distributions conditioned on these events. Throughout the paper we will informally refer to these centers as \emph{local averages}. Briefly speaking, our EDA is designed to implement a statistical sampling with the aim of approximating the local averages with high probability. At the same time, the stochastic information contained within the individuals in the final generation of the EDA allows us to estimate the entire unknown probability distribution over the vectors of delivery times. We finish the section with the following very simple lemma that plays an important motivational role in the design of our EDA.

\begin{lem}\label{convexInside}
Consider $\mathbb{R}^n$ equipped with a norm $\| \cdot \|$. Let $\epsilon > 0$ and suppose we are given a sequence of vectors $\{\vec{u}_i\}_{i=1}^n$ in $\mathbb{R}^n$ such that $\forall \, i$ and $j$ with $1 \leq i \leq j \leq n$ we have $\|\vec{u}_i - \vec{u}_j\| \leq \epsilon$. Then, given any convex combination of vectors in the sequence $\{\vec{u}_i\}_{i=1}^n$, i.e. a vector $\vec{v} = \sum_{i=1}^n r_i \vec{u}_i$ where each $r_i \in [0, \, 1]$ and $\sum_{i=1}^n r_i = 1$, the following is true: $\forall \, i \in \{1, \, 2, \ldots, n\}$ we have $\|\vec{u}_i - \vec{v}\| \leq (1-r_i) \epsilon \leq \epsilon$.\footnote{Evidently, Proposition~\ref{convexInside} is valid for an arbitrary normed vector space in place of $\mathbb{R}^n$.}
\end{lem}
\begin{proof}
Indeed, given an $i \in \{1, \, 2, \ldots, n\}$, we have $$\|\vec{u}_i - \vec{v}\| = \left \|\vec{u}_i - \sum_{j=1}^n r_j \vec{u}_j \right \| =$$$$= \left \|\sum_{j=1}^n r_j \vec{u}_i - \sum_{j=1}^n r_j \vec{u}_j \right \| = \left \|\sum_{j=1}^n r_j \left(\vec{u}_i - \vec{u}_j \right)\right \| \leq$$$$\leq \sum_{j=1}^n r_j \|\vec{u}_i - \vec{u}_j\| \leq \left(\sum_{j=1, \, j \neq i}^n r_j\right) \epsilon = (1-r_i)\epsilon \leq \epsilon,$$
which proves the lemma.
\end{proof}
\section{EDA motivated by the Theoretical Observations}\label{algorithmSect}
The design of the EDA appearing below is largely based on the novel theoretical developments in Section~\ref{EuclidEmbedSect}. First of all, select a constant $\epsilon >0$. detailed discussion regarding the choice of $\epsilon$ with respect to a probability distribution on the delivery times vectors and the desired approximation ratio of the potential solutions will be provided in the next section. For now we just include a small hint that the choice of $\epsilon$ is largely related to the side lengths of the $n$-dimensional cubes where the delivery time vectors are most likely to occur.

The search space $\Omega$ consists of three types of individuals. The \emph{regular} individuals in the search space of our EDA are sequences of vectors in $\mathbb{R}^n$, $\{\vec{q}(\theta_i)\}_{i=1}^l$, such that $\|\vec{q}(\theta_i) - \vec{q}(\theta_j)\|_{\infty} \leq \epsilon$. The idea is that the regular individuals collect samples from high-concentration events. In particular, it will allow us to estimate the ``local averages" of the probability distribution modeling the uncertainty.

Apart from regular individuals, every generation of the EDA stores a \emph{counter} individual that is an ordered pair of the form $\left(\left\{\left(\vec{q}(\lambda_i), \, k_i \right)\right\}_{i=1}^m, \, K \right)$ consisting of a sequence of ordered pairs of the form $\left\{\left(\vec{q}(\lambda_i), \, k_i \right)\right\}_{i=1}^m$ where $\vec{q}(\lambda_i) \in [0, \infty)^n$ are vectors of delivery times, while $k_i \in \mathbb{N}$ and a \emph{normalizing integer} $K = \sum_{i=1}^m k_i$. The idea is that $k_i$ keeps track of how many times the vector $\vec{q}(\lambda_i)$ has been encountered during the sampling process while $K$ is, evidently, the overall total number of sampling attempts.

In the last generation (when we stop the algorithm), every regular individual is \emph{finalized} and becomes a final individual. A \emph{final} individual is an ordered 5-tuple of the form $\left(\{\vec{q}(\theta_i)\}_{i=1}^l, \, \{t_i\}_{i=1}^l, \, \vec{q}(\overline{\theta}), \, \pi, \, r \right)$ where $\{\vec{q}(\theta_i)\}_{i=1}^l$ is a regular individual, $\{t_i\}_{i=1}^l$ is a sequence of rational coefficients in $[0, \, 1]$ such that $\sum_{i=1}^l t_i = 1$, $\vec{q}(\overline{\theta}) = \sum_{i=1}^l t_i \vec{q}(\theta_i)$ is a delivery times vector that is a convex combination of the delivery times vectors in the sequence $\{\vec{q}(\theta_i)\}_{i=1}^l$, $\pi$ is a permutation schedule and $r \in [0, 1]$ is a rational number. The main idea is that the delivery times vector $\vec{q}(\overline{\theta})$ is an estimated ``local average" of the delivery times vectors contained in the sequence representing the corresponding regular individual. The permutation schedule $\pi$ is constructed via a known polynomial-time approximation scheme (such as one in \cite{mitavskiy2012polynomial}) for the instance $\overline{\theta}$ of the single-machine scheduling problem. As discussed in the previous section, thanks to Lemma~\ref{LipshitzNormProp}, the permutation schedule $\pi$ is a ``sufficiently good" approximation schedule for all the instances falling within the $\epsilon$-neighborhood of the instance $\overline{\theta}$. According to Proposition~\ref{convexInside}, all of the instances in the sequence contained within the corresponding regular individual are within the $\epsilon$-neighborhood of the delivery times vector $\vec{q}(\overline{\theta})$ so that, in particular, the permutation-schedule $\pi$ is sufficiently good for these instances. In other words, the permutation schedule $\pi$ may be considered to robust.

For a wide class of probability distributions that have polynomially many events happening with high concentration, as discussed informally in the previous section and will be described rigorously in the next section, after a polynomially sized runtime before finalizing, a randomly sampled individual will fall within the $\epsilon$-neighborhood of the delivery times vector $\vec{q}(\overline{\theta})$ of one of the final individuals with high probability. The rational number $r \in [0, 1]$ represents an estimated probability that a randomly sampled vector of delivery times $\vec{q}(\lambda)$ falls within the $\epsilon$ neighborhood of the delivery times vector $\vec{q}(\overline{\theta})$ while each $t_i$ is an estimated conditional probability of the delivery times vector $\vec{q}(\theta_i)$ with respect to the high-concentration event the estimated local average of which is the delivery times vector $\vec{q}(\overline{\theta})$. Thereby, the population in the last generation is meant to contain all the necessary information to estimate the entire probability distribution over the vectors of delivery times. We now proceed with the detailed exposition. The Description of the EDA as well as a detailed description of every one of its subroutines (namely, theoretically guided initialization, mutation, and finalizing regular individuals) appears in Algorithm~\ref{algorithm}.

\begin{algorithm}
\caption{EDA for Single Machine Scheduling with Random Delivery Times}
\label{algorithm}
\begin{algorithmic}[1]
 \STATE Initialize a population $Pop$ as follows: First, set the population size to $2$. Sample a vector of delivery times $\vec{q}(\lambda)$. Initialize a counter individual $I_0$ and a regular individual $I_1$  as described in the subroutines below and let $Pop = \{I_0, \, I_1\}$.

 \WHILE {The total number of generations is smaller than $T + 1 \in \mathbb{N}$}
 \STATE Sample a vector of delivery times $\vec{q}(\lambda)$ with respect to the probability distribution over the delivery times vectors.
 \STATE Set a counting index $j=1$
 \WHILE{$0 \leq j < L$}
 \STATE Mutate the individual $I_j$ with respect to the instance $\lambda$ as described in the subroutines above.
 \ENDWHILE
 \IF{(none of the individuals has been mutated (more precisely, if every individual $I_j$ for $0 \leq j < L$ returns a boolean value ``false")}
 \STATE Increment the population size $L := L+1$ by one and initialize a regular individual $I_{L-1}$ with respect to the vector of delivery times $\vec{q}(\lambda)$ and add it to the population.
 \ENDIF
 \ENDWHILE
 \STATE Finalize everyone of the regular individuals according to the subroutine described above and output the population $Pop$.
\end{algorithmic}
\end{algorithm}

The subroutines of the algorithm are described as follows:
\begin{itemize}

\item
\textbf{Initialization of a regular individual with respect to a vector $\vec{q}(\lambda)$ of delivery times:} Given a vector $\vec{q}(\lambda)$ of delivery times for a sampled instance of the single-machine scheduling problem of size $n$, set $l = 1$ and let the new individual be the sequence $\{\vec{q}(\theta_i)\}_{i=1}^l$ where $\theta_1 = \lambda$.

\item
\textbf{Initialization of the counter individual with respect to a vector $\vec{q}(\lambda)$ of delivery times:} Given a vector $\vec{q}(\lambda)$ of delivery times for a sampled instance of the single-machine scheduling problem of size $n$, set $m = 1$ and $k_1 = 1$ and let the counter individual be the sequence $\{(\vec{q}(\theta_i), \, k_1)\}_{i=1}^l$ where $\theta_1 = \lambda$.

\item
\textbf{Mutation of a regular individual with respect to a vector $\vec{q}(\lambda)$ of delivery times:} Given a regular individual $\{\vec{q}(\theta_i)\}_{i=1}^l$ and a vector $\vec{q}(\lambda)$ of delivery times for a sampled instance of the single-machine scheduling problem of size $n$,  the mutation of  a regular individual is described as follows.

\begin{algorithmic}
\STATE set a boolean variable $mutation = false$;
\IF{ $\text{ for some } i \in \{1, \, 2, \ldots, l\}$ the following is true: $\left( \left[\vec{q}(\theta_i) = \vec{q}(\lambda) \right] OR \left[\|\vec{q}(\theta_i) - \vec{q}(\lambda)\|_{\infty} > \epsilon \right] \right)$} \STATE do nothing;

\ELSE
\STATE increment $l := l+1$ by one and let $\vec{q}(\theta_l) = \vec{q}(\lambda)$ and set $mutation = true$;
\ENDIF
\RETURN $mutation$.
\end{algorithmic}

\item
\textbf{Mutation of the counter individual with respect to a vector $\vec{q}(\lambda)$ of delivery times:} Given the counter individual of the form $\left(\left\{\left(\vec{q}(\lambda_i), \, k_i \right)\right\}_{i=1}^m K\right)$ and a vector $\vec{q}(\lambda)$ of delivery times for a sampled instance of the single-machine scheduling problem of size $n$, the mutation of   the counter individual is described as follows.

\begin{algorithmic}
\STATE set a boolean variable $addOrNot = false$.;
\IF{ $\left( \vec{q}(\lambda_i) = \vec{q}(\lambda) \text{ for some } i \in \{1, \, 2, \ldots l\} \right)$ }
\STATE increment $k_i := k_i + 1$ by one;

\ELSE
\STATE increment $m := m + 1$ by one and add the ordered pair $\left(\vec{q}(\lambda_m), \, k_m \right)$ where $\vec{q}(\lambda_m) = \vec{q}(\lambda)$ and $k_m = 1$ to the sequence representing the counter individual. reset the boolean value $addOrNot := true$;
\ENDIF

\STATE increment $K := K+1$ by one;

\RETURN $addOrNot$.
\end{algorithmic}

In the case the returned boolean value after mutation of a regular or a counter individual is ``true" we will say that the individual \emph{has been mutated}. Every non-terminal population in our EDA contains a unique counter individual and a variable number of regular individuals. Once we decide to stop running the EDA, every one of the regular individuals is finalized according to the subroutine described below.

\item
\textbf{Finalizing a regular individual in a population:} Suppose we are given a population containing the counter individual $count = \left(\left\{\left(\vec{q}(\lambda_i), \, k_i \right)\right\}_{i=1}^m, \, K \right)$ and a regular individual $regular = \{\vec{q}(\theta_i)\}_{i=1}^l$. For every vector of delivery times $\vec{q}(\theta_i)$ find the ordered pair $\left(\vec{q}(\lambda_{j(i)}), \, k_{j(i)} \right)$ within the $count$ individual such that $\vec{q}(\theta_i) = \vec{q}(\lambda_{j(i)})$. (From the construction of the EDA above, it will be apparent that there always exists a unique ordered pair with this property in the counter individual.) Set the normalizing constant $N = \sum_{i=1}^l k_{j(i)}$. For every $i \in \{1, \, 2, \ldots, l\}$ let $$t_i = \frac{k_{j(i)}}{N}; \; \vec{q}(\bar{\theta}) = \sum_{i=1}^l t_i \cdot \vec{q}(\theta_i) \text{ and } r = \frac{N}{K}.$$ Construct a permutation schedule $\pi$ for the instance $\bar{\theta}$ according to a polynomial-time approximation scheme (\cite{mitavskiy2012polynomial}) with an appropriately selected approximation ratio (the details will be discussed in the next section). Form the corresponding final individual $\left(\{\vec{q}(\theta_i)\}_{i=1}^l, \, \{t_i\}_{i=1}^l, \, \vec{q}(\overline{\theta}), \, \pi, \, r \right)$.
\end{itemize}

In the upcoming section we are going to study the types of probability distributions on $\mathbb{R}^n$ for which the EDA described above performs well. If one has some heuristic knowledge or, at least a guess, about the basic properties of an ``algorithm-pleasable" probability distribution, it is not hard to tune the approximation-related parameter $\epsilon$, and the total runtime $T$ before finalizing the regular individuals so that the population in the last generation achieves the following objectives stated informally below:
\begin{itemize}

\item \textbf{Objective 1:} With high probability a randomly sampled vector of delivery times $\vec{q}(\gamma)$ will fall within the $\epsilon$-neighborhood of the delivery times vector $\vec{q}(\overline{\theta})$ of at least one final individual $\left(\{\vec{q}(\theta_i)\}_{i=1}^l, \, \{t_i\}_{i=1}^l, \, \vec{q}(\overline{\theta}), \, \pi, \, r \right)$ contained within the population in the last generation. Provided that the schedule $\pi$ has been constructed within an appropriate approximation ratio for the instance $\overline{\theta}$ and the approximation parameter $\epsilon$ has been selected wisely, Proposition~\ref{LipshitzNormProp} will then tell us that the schedule $\pi$ also provides a solution for the instance $\gamma$ within a desirable approximation ratio. Notice that computing $\|\vec{q}(\gamma) - \vec{q}(\theta_i)\|_{\infty}$ can be carried out in linear time (in fact, in constant time if one has sufficiently many parallel processors).

\item \textbf{Objective 2:} The stochastic information contained within the final individuals in the population at the last generation allows us to estimate the entire probability distribution over the vectors of delivery times very much as follows. Let's say the population at the last generation is of the form
\begin{equation}\label{finalPopExEq}
\text{Pop} = (I_0, \, I_1, \, I_2, \ldots, I_{L})
\end{equation}
where $I_0$ is a counter individual of the form $I_0 = \left(\left\{\left(\vec{q}(\lambda_i), \, k_i \right)\right\}_{i=1}^m, \, K \right)$ and, for $1 \leq j \leq L$, the final individuals, $I_j$, of the form $I_j = \left(\{\vec{q}(\theta_i^j)\}_{i=1}^{l(j)}, \, \{t_i^j\}_{i=1}^{l(j)}, \, \vec{q}(\overline{\theta}_j), \, \pi_j, \, r_j \right)$. For $j \in \{1, \, 2, \ldots, L\}$, consider now the event $E_j$ that $\forall \, i \in \{1, \, 2, \ldots, l(j)\}$ a randomly sampled vector of delivery times $\vec{q}(\lambda)$ falls within the $\epsilon$-neighborhood of the vector $\vec{q}(\theta_i^j)$. Then $Pr(E_j) \approx r_j$. In other words, the probability that a randomly sampled vector of delivery times $\vec{q}(\lambda)$ falls within the $\epsilon$-neighborhood of the delivery times vector $\vec{q}(\overline{\theta}_j)$ contained within the individual $I_j$ is well-approximated by $r_j$. Furthermore, the vector $\vec{q}(\overline{\theta}_j)$ is very close to the mean (or expectation) of the multivariate conditional distribution over the delivery times vectors conditioned on the event $E_j$. Consider now the discrete multivariate distribution on the $\epsilon$-neighborhood of the vector $\vec{q}(\overline{\theta}_j)$, call it $CondDistrib(j)$, which assigns the probability $t_i^j$ to the vector of delivery times $\vec{q}(\theta_i^j)$. We then hope to deduce that the multivariate conditional probability distribution over the vectors of delivery times conditioned on the event $E_j$, is well-approximated by the discrete multivariate distribution $CondDistrib(j)$ on the $\epsilon$-neighborhood of the vector $\vec{q}(\overline{\theta}_j)$ via an appropriate statistical parameter fitting. Probably the simplest parameter fitting would be to assume that the conditional distributions are reasonably close to Gaussians. In this case all we need to estimate is the covariance matrix, and this is easily done by computing the covariance matrix of the corresponding multivariate distribution  $CondDistrib(j)$.
\end{itemize}

\section{Theoretical Analysis of the Poposed EDA}\label{ProbDistributionAnalysisSect}
\subsection{Mathematical Framework, Assumptions on the Probability Distributions over the Vectors of Delivery Times and Interpretive Discussion}\label{introToTechnicalSubsect}
Continuing with the notation in the previous section, we will assume that the multivariate probability distribution $D_n$ on the vectors of delivery times for an instance of size $n$ of the single-machine scheduling problem satisfies the following properties. First of all, since delivery times are always non-negative, it makes sense that the corresponding probability distributions are concentrated in the non-negative semi-space, $\prod_{i=1}^n [0, \infty)$. Furthermore, the following assumption on the probability distribution $D_n$ will be made.

\begin{itemize}
\item \textbf{Assumption 1:} For technical reasons, we will also assume that the delivery times can not be uncontrollably large in the sense that $\exists$ a polynomially growing sequence $\{M_n\}_{n=1}^{\infty} \in [0, \infty)$ (in other words, $M_n = O(n^d)$ for some $d \in \{0, \infty\}$) such that $q_i \leq M_n$ $\forall \, i \in \{1, \, 2, \ldots, n\}$ with probability $1$. For the sake of concreteness, let's say that $M_n \leq const_2 \cdot n^d$ for some $const_2 \in (0, \infty)$.

\item
\textbf{Assumption 2:} We assume that for the search space of instances of the single-machine scheduling problem of size $n$, there are positive bounds $\epsilon > 0$ and $\beta > 0$ and polynomially many events $E_1, \, E_2, \ldots E_{f(n)}$ (i.e. $f(n) = O(n^c)$ for some constant $c \in [0, \infty)$ independent of $n$) such that each of the events $E_i$ is contained within an $n$-dimensional cube $E_i \subseteq \prod_{j=1}^n [e_j^i - \frac{\epsilon}{2}, \, e_j^i + \frac{\epsilon}{2}]$ having side lengths of at most $\epsilon$ and, with high probability, for the sake of concreteness, let's say, with probability at least $1 - O\left(f(n) \cdot \exp({-\beta \cdot n})\right)$, a randomly sampled vector of delivery times $\vec{q}(\lambda) \in \bigcup_{i=1}^{f(n)} E_i$ lies within at least one of the events $E_i$. Let $const_1 \in (0, \infty)$ be selected so that for all sufficiently large $n$ we have $f(n) \leq const_1 \cdot n^c$. Finally, we also assume that all of the events $E_i$ have asymptotically similar probabilities.\footnote{Notice that $\beta$ does not depend on the problem size $n$} More precisely $\min_{i=1}^{f(n)}Pr(E_i) = \Omega\left(\frac{1}{f(n)}\right)$. In other words, once again, we can select a constant $const > 0$ so that $\min_{j = 1}^n Prob(E_j) \leq \frac{const}{f(n)}$.

 \end{itemize}

\begin{rem}\label{assumptionRem}
Notice that the assumption above immediately implies that $$\sum_{i=1}^{f(n)}Pr(E_i) \geq Pr\left(\bigcup_{i=1}^{f(n)} E_i \right) =$$$$= 1 - Pr\left(\overline{\bigcup_{i=1}^{f(n)} E_i} \right) \geq 1 - O\left(f(n) \cdot \exp({-\beta \cdot n})\right).$$
\end{rem}

Informally speaking, the assumption above says that the probability distribution over the vectors of delivery times is highly concentrated within the union of polynomially many cubes the side lengths of which are bounded by a constant (the constant may depend on the problem size $n$ which is also the dimension of the delivery times vectors and the cubes). As suggested within objective 2 at the end of the previous section, the population in the last generation is expected to contain at least $f(n)$ final individuals: one individual for every event $E_i$. The counter individual will help us to estimate the actual probabilities of every one of the events $E_i$. Moreover, together with the information contained in the final individual corresponding to an event $E_i$, we can also estimate the conditional probability distribution conditioned on the event $E_i$. The remainder of the current section is devoted to answering the following crucial questions.

\begin{itemize}
\item \textbf{Question 1:} For how long do we need to run the EDA prior to finalizing the individuals?

\item \textbf{Question 2:} What is the quality of the permutation-schedules we obtain in terms of the approximation ratios?
\end{itemize}

The main theorem of the current section that addresses the two questions of central importance stated above is formulated quantitatively according to the following scheme: ``Select a small constant $\delta > 0$. Under the assumptions 1 and 2 above, given that we run the EDA presented in Section~\ref{algorithmSect} for at least polynomially many time steps with respect to the problem size $n$, but not with respect to $\delta$, $T(n, \, \delta)$, prior to finalizing the individuals, the probability that certain undesirable events $U_1$ or $U_2(\delta)$ take place is exponentially small...". In order to alleviate the complexity of formal presentation, we define the undesirable events $U_1$ and $U_2(\delta)$ below.

\begin{defn}\label{auxiliary1Defn}
Suppose that the search space of instances of the single machine scheduling problem of size $n$ satisfies assumption 2 above. Let $T \in \mathbb{N}$ and suppose that we run the EDA described in Section~\ref{algorithmSect} for some $t \geq T$ time steps prior to finalizing the regular individuals. We say that an \emph{event} $E_j$ \emph{corresponds to a final individual of the form} $\left(\{\vec{q}(\theta_i)\}_{i=1}^l, \, \{t_i\}_{i=1}^l, \, \vec{q}(\overline{\theta}), \, \pi, \, r \right)$ if $\forall \, i \in \{1, \, 2, \ldots, l\}$ $\vec{q}(\theta_i) \in E_j$. The undesirable event $U_1$ is then defined to be the event that at least one of the events $E_j$ does not correspond to any of the final individuals. Equivalently, $U_1$ is the event that $\exists \, j \in \{1, \, 2, \ldots f(n)\}$ such that neither one of the final individuals corresponds to the event $E_j$. Now let $\delta > 0$. We define $U_2(\delta)$ to be the event that if a final individual corresponding to an event $E_j$ is of the form $\left(\{\vec{q}(\theta_i)\}_{i=1}^l, \, \{t_i\}_{i=1}^l, \, \vec{q}(\overline{\theta}), \, \pi, \, r \right)$, then the $\| \cdot \|_{\infty}$ of the difference between an estimated average of the conditional probability distribution over the event $E_i$, $\vec{q}(\overline{\theta})$, and the actual average of the conditional probability distribution over the event $E_i$, call it $\vec{\mu}_i$, is at least $\delta$: i.e. $\|\overline{\theta} - \vec{\mu}_i\|_{\infty} \geq \delta$.
\end{defn}

Clearly, the complement of the undesirable event $U_1$ corresponds to achieving objective 1 while the complement of the undesirable event $U_2(\delta)$ is necessary to achieve objective 2, stated at the end of the previous section. We are now ready to state and to interpret the main result of the current section.
\subsection{The Statement and Interpretation of the Main Theorem}\label{stateInterpretSubsect}
The following result addresses Questions 1 and 2 stated in the previous subsection.

\begin{thm}\label{EstDistribObjectiveLem}
Suppose the EDA described in Section~\ref{algorithmSect} runs on instances of the single-machine scheduling problem of size $n$ in uncertain environments satisfying assumptions 1 and 2 in Section~\ref{ProbDistributionAnalysisSect}. Recall from the assumptions 1 and 2 that the upper bound on the maximal length of every delivery time vector with respect to $\| \cdot \|_{\infty}$, $M_n \leq const_2 \cdot n^d$, $\min_{j = 1}^{f(n)} Prob(E_j) \leq \frac{const}{f(n)}$ and the total number of the events $E_j$ happening with substantial probability is $f(n) \leq const_1 \cdot n^c$. Select constants $\delta$ and $\alpha \in (0, 1)$ and $l \in (0, \, \infty)$.
Now let
\begin{equation}\label{runtimeDefnEq}
T = \frac{const_1}{(1-\alpha) \cdot const} \cdot n^{2d +l + c}.
\end{equation}
Suppose we run the EDA described in Section~\ref{algorithmSect} for a time $t \geq T$ and then finalize all of the regular individuals in the population. Then the probability of the union $U = U_1 \cup U_2$ of the two events $U_1$ and $U_2(\delta)$ introduced in Dection~\ref{auxiliary1Defn}, is at most $$const_1 \cdot n^c \exp \left( -n^{2d + l} \frac{\alpha^2}{2(1-\alpha)}\right) + 2n \exp\left(n^{-\frac{2 \delta^2}{(const_2)^2}l}\right).$$
Furthermore, consider the final individual of the form $\left(\{\vec{q}(\theta_i)\}_{i=1}^l, \, \{t_i\}_{i=1}^l, \, \vec{q}(\overline{\theta}), \, \pi, \, r \right)$ corresponding to the event $E_j$ in the sense of Dection~\ref{auxiliary1Defn}. Suppose that the permutation schedule $\pi$ has been constructed for the instance of the single-machine scheduling problem via a polynomial-time approximation scheme such as in \cite{mitavskiy2012polynomial}. Then $\forall \, $ instance of the single-machine scheduling problem $\rho$ such that $\vec{q}(\rho) \in E_j$ we have
\begin{equation}\label{ApproxQualityBoundEq}
\frac{J(\rho)_{\pi}^{\max}}{J(\rho)^*} \leq \frac{J(\overline{\theta})_{\pi}^{\max} + \epsilon + \delta }{\frac{1}{r}J(\overline{\theta})_{\pi}^{\max} - \epsilon - \delta}.
\end{equation}
\end{thm}

Since the bounds in Theorem~\ref{EstDistribObjectiveLem} involve a number of constants, we provide a discussion regarding trade-offs between selecting smaller or larger constants. First of all, recall that $c$ and $d$ are degrees of the polynomials that bound the the total number of the concentration events of the probability distributions $D_n$ and the largest possible size of the delivery times. Of course, prefer them to be as small as possible, but they are fully dependent on the specific application and expert knowledge of the circumstances related to the application. Similar remarks pertain to the constants $const$, $const_1$, $const_2$ and $\epsilon$. Clearly we prefer $const_1$ and $const_2$ to be as small as possible, while $const$ to be as large as possible, but once again, this is dependent on a specific application. $\epsilon$ is the parameter related to the concentration properties of the conditional distribution on the vector of delivery times. Notice that it has no direct effect on the runtime bounds in the sense that it does not appear in the run-time defining (\ref{runtimeDefnEq}) within the statement of Theorem~\ref{EstDistribObjectiveLem}. It certainly effects the quality of the approximation ratio directly as stated in (\ref{ApproxQualityBoundEq}).\footnote{A careful reader will notice that the parameter $\epsilon$ is related to the total number of the high-concentration events $E_i$ (see assumption 2 at the beginning of the current section). Depending on the expert knowledge about the properties of the probability distribution that models uncertainty, it may then influence the parameter $c$ that does effect the algorithm's runtime bound in (\ref{runtimeDefnEq}).} On the other hand, parameters $\alpha \in (0, \, 1)$, $\delta > 0$ and $l>0$ must be selected to preserve a trade-off between the runtime complexity and the quality of the guaranteed bounds on the approximation ratios of the schedules and the probability of successfully achieving these bounds vs. the computational expense of the runtime complexity. Indeed, from the bounds in   Theorem~\ref{EstDistribObjectiveLem} it is clear that the larger is $l$ the longer is the runtime $T$ but the smaller is the probability of the undesirable events. Similarly, the larger is $\alpha \in (0, \, 1)$ the longer is the runtime and the smaller is the the probability of the undesirable events.

The remaining subsection is devoted to the proof of Theorem~\ref{EstDistribObjectiveLem}.
\subsection{Establishing Theorem 3}\label{proofSubsect}
First of all, we notice that  ``Question 1" formulated in the previous section can be alternatively restated in terms of a mildly extended ``coupon collector problem" (see, for instance, \cite{blom1994problems}, \cite{erdos1961classical} and \cite{auger2011theory}) as follows:
\begin{itemize}
\item \textbf{Question 1a:} What is the waiting time until we sample \emph{sufficiently many elements} in every one of the events $E_i$ to estimate its conditional distribution up to a \emph{satisfactory criteria}?
\end{itemize}

Recall from Subsection~\ref{robustLinkSubsect}  that we aim to estimate the local average i.e. the mean of each conditional distribution conditioned on the event $E_i$. This can be achieved quite well thanks to a rather simple corollary from the classical Hoeffding inequality (see \cite{auger2011theory}). In fact, considering the assumption on the probability distribution $D_n$, this very simple satisfactory criteria is already quite powerful in view of its connection to answering Question 2. Indeed, thanks to Proposition~\ref{convexInside}, the mean of such a conditional distribution is at most ``$\epsilon$-far" from any other delivery times vector in the event $E_i$. Thus, thanks to Proposition~\ref{LipshitzNormProp}, the permutation schedule $\pi$ for the estimated average vector of delivery times $\vec{q}(\overline{\theta})$ contained within the final individual corresponding to the event $E_i$ will provide a solution within an approximation ratio that depends on $\epsilon$ as well as on the approximation ratio and the value $J(\overline{\theta})_{\pi}^{\max}$ of the permutation-schedule $\pi$ constructed for the instance $\theta$. We now proceed with a detailed analysis.

We start with establishing the following corollary of the classical Hoeffding inequality that's well-suited in order to understand how many samples we want to obtain inside each of the events $E_i$ prior to stopping the algorithm.
\begin{thm}\label{Hoeffding}
Suppose $\vec{X}_1, \, \vec{X}_2, \ldots, \vec{X}_k, \ldots$ is any sequence of bounded i.i.d. $\mathbb{R}^n$-valued random variables having the common expectation $\vec{\mu}$ and where $\|\vec{X}_i\|_{\infty} \leq M$ almost surely (i.e. with probability $1$). Select $\delta > 0$. Then, $\forall \, k \in \mathbb{N}$ we have
\begin{equation}\label{HoeffdingIneqalEq}
Pr\left(\left\|\frac{\sum_{i=1}^k \vec{X}_i}{k} - \vec{\mu}\right\|_{\infty} \geq \delta \right) \leq 2 n \cdot \exp\left(- 2k\frac{\delta^2}{M^2}\right)
\end{equation}
\end{thm}
\begin{proof}
For a $j \in \{1, \, 2, \ldots n\}$ consider the random sequence of the $j^{\text{th}}$ coordinates of the vectors $\vec{X}_k$, $\{X_k^j\}_{k=1}^{\infty}$. Clearly, this is a sequence of i.i.d real-valued random variables taking values in the interval $[0, \, M]$ with a common mean (expectation) $\mu^j$. Of course, $\mu^j$ is the $j^\text{th}$ coordinate of the vector $\mu$. A version of Hoeffding inequality presented in Theorem 1.11 of \cite{auger2011theory} tells us that $\forall \, r>0$ we have
$$Pr\left(\sum_{i=1}^k X_i^j \geq \sum_{i=1}^k E(X_i^j) + r = k \cdot \mu_j + r \right) \leq$$
\begin{equation}\label{HoeffdingCoordinateInProofUpperIneq}
\leq \exp\left(-2\frac{r^2}{\sum_{i=1}^k M^2}\right) = \exp \left(-2\frac{r^2}{k \cdot M^2}\right)
\end{equation}
and, likewise,
\begin{equation}\label{HoeffdingCoordinateInProofLowerIneq}
Pr\left(\sum_{i=1}^k X_i^j \leq k \cdot \mu_j - r \right) \leq \exp \left(-2\frac{r^2}{k \cdot M^2}\right).
\end{equation}
Combining (\ref{HoeffdingCoordinateInProofUpperIneq}) and (\ref{HoeffdingCoordinateInProofLowerIneq}) and observing that the probability of the union of two events is bounded above by the sum of the corresponding probabilities, we deduce that $\forall \, j \in \{1, \, 2, \ldots, n\}$ we have
$$Pr\left(\left|\sum_{i=1}^k X_i^j - k \cdot \mu_j \right| \geq r \right) \leq 2\exp \left(-2\frac{r^2}{k \cdot M^2}\right)$$ or, equivalently,
\begin{equation}\label{HoeffdingCooordinateInProofDualSidedEneq}
Pr\left(\left|\frac{\sum_{i=1}^k X_i^j}{k} - \mu_j \right| \geq \frac{r}{k} \right) \leq 2\exp \left(-2\frac{r^2}{k \cdot M^2}\right).
\end{equation}
Applying (\ref{HoeffdingCooordinateInProofDualSidedEneq}) with $r = \delta \cdot k$, we deduce that
\begin{equation}\label{HoeffdingCooordinateInProofDualSidedEneqSubst}
Pr\left(\left|\frac{\sum_{i=1}^k X_i^j}{k} - \mu_j \right| \geq \delta \right) \leq 2\exp \left(-2k\frac{\delta^2}{M^2}\right).
\end{equation}
Notice now that, according to the definition of $\| \cdot \|_{\infty}$, the event
$$U = \left\{\omega \, | \, \left\|\frac{\sum_{i=1}^k \vec{X}_i(\omega)}{k} - \vec{\mu}\right\|_{\infty} \geq \delta\right\}$$
the probability of which we aim to bound above, is the union of the events
$$U_j = \left\{\omega \, | \, \left|\frac{\sum_{i=1}^k X_i^j(\omega)}{k} - \mu_j \right| \geq \delta \right\}$$ over all coordinates $1, \, 2, \ldots, n$. Thus, the probability of the event $U$ is bounded above by the sum of $n$ identical bounds on the right hand side of (\ref{HoeffdingCooordinateInProofDualSidedEneqSubst}) immediately implying (\ref{HoeffdingIneqalEq}).
\end{proof}

Once an appropriate $\delta>0$ and $k$ have been selected to apply (\ref{HoeffdingIneqalEq}), the remaining part of the analysis involves estimating the total number of generations (or samples) necessary to collect at least $k$ vectors of delivery times in every one of the $E_i$s. In other words, the remaining part of the analysis boils down to a version of the coupon collector problem. Indeed, recall from the assumption 2, that there are polynomially many events $E_1, \, E_2, \ldots E_{f(n)}$ such that $minProb(n) = \min_{i=1}^{f(n)} Pr(E_i) = \frac{const}{f(n)}$. Suppose now we were to wait until $k$ samples of delivery times vectors $\vec{q}(\theta_i)$ have been collected in each of the events $E_i$. This waiting time random variable is certainly bounded above by the random variable $\overline{T}$, for which the probability of obtaining a sample of a delivery times vector inside of an event $E_i$ is $minProb(n)$, while the probability of sampling a delivery times vector outside $\bigcup_{i=1}^n E_i$ is $1 - f(n) \cdot minProb(n)$. For every specified $i \in \{1, \, 2, \ldots, f(n)\}$, the total number of samples of delivery times vectors inside the event $E_i$ after $\overline{T}$ time steps is distributed binomially with success probability $minProb(n) = \frac{const}{f(n)}$. Thus, according to the classical Chernoff bound (see, for example, part $a)$ of corollary 1.10 of \cite{auger2011theory}, we deduce that for a constant $\alpha \in (0, \, 1)$, after $\overline{T} = \frac{k \cdot f(n)}{(1-\alpha) \cdot const}$ time steps, the probability
$$Pr\left(\# \text{of delivery time vectors }\in E_i \text{ after } \overline{T} \text{ time steps} \leq k \right)$$
\begin{equation}\label{ChernoffInequalApplyInequal}
\leq \exp\left(-k \frac{\alpha^2}{2(1-\alpha)}\right).
\end{equation}
Indeed, the mean of the binomial distribution of sampling delivery time vectors from the set $E_i$ after $\overline{T} = \frac{k \cdot f(n)}{(1-\alpha) \cdot const}$ attempts with success probability $minProb(n) = \frac{const}{f(n)}$ is $minProb(n) \cdot \overline{T} = \frac{const}{f(n)} \cdot \frac{k \cdot f(n)}{(1-\alpha) \cdot const} = \frac{k}{1-\alpha}$ so that the traditional Chernoff inequality stated in corollary 1.10 part a) \cite{auger2011theory} with $\delta = \alpha$ entails the bound in (\ref{ChernoffInequalApplyInequal}). The event $Ev(k)$ that at least one of the $E_j$s contains fewer than $k$ samples of the delivery times vectors after the time $\overline{T} = \frac{k \cdot f(n)}{(1-\alpha) \cdot const}$ is the union over $i \in \{1, \, 2, \ldots, f(n)\}$ of the corresponding events that the event $E_i$ contains fewer than $k$ samples of the delivery times vectors. Thus the probability of the event $Ev(k)$ is bounded above by the sum of the $f(n)$ probabilities estimated with a common bound $\exp\left(-k\frac{\alpha^2}{2(1-\alpha)}\right)$ in (\ref{ChernoffInequalApplyInequal}). We now deduce the following important lemma.

\begin{lem}\label{couponCollectlem}
Consider now the EDA described in Section~\ref{algorithmSect} running on the instances of the single machine scheduling problem of size $n$ with the parameter $\epsilon$ as in the assumption 2 in Section~\ref{ProbDistributionAnalysisSect}. Select a constant $\alpha > 0$ and $k \in \mathbb{N}$. Recall from the assumptions 1 and 2 at the beginning of Subsection~\ref{introToTechnicalSubsect} that $const_1 \in (0, \infty)$ has been selected so that for all sufficiently large $n$ we have $f(n) \leq const_1 \cdot n^c$ Suppose that we have run our EDA for at least $t \geq \frac{k \cdot const_1 \cdot n^c}{(1-\alpha) \cdot const}$ generations prior to finalizing the regular individuals (recall that $const$ is selected so that $minProb(n) = \frac{const}{f(n)}$). Then the probability that the undesirable event $U_1$ described in Dection~\ref{auxiliary1Defn} takes place, or at least one of the regular individuals contains fewer than $k$ samples of delivery time vectors counting multiplicities (recall that the multiplicities are recorded within the counter individual), is at most $const_1 n^c \exp\left(-k\frac{\alpha^2}{2(1-\alpha)}\right)$.
\end{lem}

Theorem~\ref{Hoeffding} tells us how many samples (counting multiplicities recorded within the counter individual) we need to obtain within everyone of the regular individuals to approximate the means of the conditional distributions with respect to the events $E_i$ up to a sufficiently small error with high probability. In turn, Lemma~\ref{couponCollectlem} informs us how long do we need to run the EDA to obtain at least the desired number of samples counting multiplicities within everyone of the regular individuals with large probability. Substituting $M = M_n$ and $k = n^{2d + l}$ into the statements of Theorem~\ref{Hoeffding} and Lemma~\ref{couponCollectlem} entails the first conclusion of Theorem~\ref{EstDistribObjectiveLem}.

Our first objective concerns the approximation ratio of the permutation-schedule $J(\overline{\theta})_{\pi}^{\max}$ on the instances of the single-machine scheduling problem $\rho$ the delivery time vectors of which are within $\epsilon$-neighborhood of the vector $\vec{q}(\overline{\theta})$. It only remains now to convert additive bounds into traditional approximation ratios (recall the discussion at the end of in Subsection~\ref{formalDescribeSubsect}) based on Propositions~\ref{LipshitzNormProp} and~\ref{convexInside}. Indeed, according to Proposition~\ref{LipshitzNormProp} together with the triangle inequality, $\forall \, \rho \in E_j$, we deduce that
$$\|\vec{q}(\overline{\theta}) - \vec{q}(\rho)\|_{\infty} \leq \|\vec{q}(\overline{\theta}) - \vec{q}(\mu_j)\|_{\infty} + \|\vec{q}(\mu_j) - \vec{q}(\rho)\|_{\infty} \leq$$
\begin{equation}\label{inFinalProofEq1}
\leq \epsilon + \delta.
\end{equation}
where $\mu_j$ is the instance of the single-machine scheduling problem of size $n$ determined by the local average of the delivery time vectors in $E_j$, $\vec{q}(\mu_j)$.
But then, according to Proposition~\ref{LipshitzNormProp} together with the triangle inequality, we deduce that $$\left|J(\overline{\theta})^{*} - J(\rho)^*\right| \leq \left|J(\overline{\theta})^* - J(\mu_j)^{*}\right| + \left|J(\mu_j)^{*} - J(\rho)^*\right| \leq $$
\begin{equation}\label{absValueDifferenceBound}
\leq \epsilon + \delta.
\end{equation}
and, likewise,
\begin{equation}\label{absValueDifferenceBoundPermute}
\left|J(\overline{\theta})_{\pi}^{\max} - J(\rho)_{\pi}^{\max}\right| \leq \epsilon + \delta.
\end{equation}
Recall that the permutation-schedule $\pi$ has been constructed on the instance of the single-machine scheduling problem $\overline{\theta}$ with the approximation ratio $r$ so that
\begin{equation}\label{approxRatioInProofEnequal}
\frac{J(\overline{\theta})_{\pi}^{\max}}{J(\overline{\theta})^{*}} \leq r \Longrightarrow J(\overline{\theta})^{*} \geq \frac{1}{r}J(\overline{\theta})_{\pi}^{\max}.
\end{equation}
But then, from (\ref{absValueDifferenceBound}), (\ref{absValueDifferenceBoundPermute}) and (\ref{approxRatioInProofEnequal}), it follows that
\begin{equation}\label{lastApproxRatioInequalProof1}
J(\rho)^{*} \geq J(\overline{\theta})^* + \epsilon + \delta \geq \frac{1}{r}J(\overline{\theta})_{\pi}^{\max} - \epsilon - \delta
\end{equation}
and
\begin{equation}\label{lastApproxRatioInequalProof2}
J(\rho)_{\pi}^{\max} \leq J(\overline{\theta})_{\pi}^{\max} + \epsilon + \delta
\end{equation}
so that the desired approximation ratio is obtained via dividing (\ref{lastApproxRatioInequalProof2}) by (\ref{lastApproxRatioInequalProof1}).
\section{Conclusions and Future Work}\label{concludeSect}
As mentioned in the introduction, most of the existing research for tackling NP-hard optimization problems in uncertain environments concentrates on finding a single robust solution: in other words, a solution which does not deteriorate much subject to random changes in the environment. In contrast to the previous work, we observe a rigorous theoretical link between the notion of robustness for the single-machine scheduling problem with uncertain delivery times and  the infinity norm defined as the maximum of the absolute values of the coordinates of a vector. Based on this novel finding together with a few other theoretical observations such as Lemma~\ref{convexInside} and Theorem~\ref{EstDistribObjectiveLem} we have designed an EDA to implement a statistical sampling procedure from an uncertain multivariate probability distribution over the delivery times.

The individuals in the final population of the EDA contain ordered tuples of the form
{\small $$(a \, vector \, of \, delivery \, times, \; schedule, \; other \, information)$$ }
(it has been explained in the article how a collection of delivery times is regarded as an $n$-dimensional vector in $\mathbb{R}^n$ having positive coordinates). We demonstrate theoretically, that if the multivariate probability distribution on the $n$-dimensional vectors of delivery times has polynomially many ``high-concentration events", then, after polynomially many time steps, the pairs $(vector \, of \, delivery \, times, \, schedule)$ contained within the individuals in the final population, produce a polynomially large ``fishnet cover" of the probability space over the delivery times vectors in the following sense. Given a randomly sampled vector of delivery times, with high probability one can find an individual contained within the final population that contains a vector of delivery times that is sufficiently close to the randomly sampled one (closeness is measured in terms of the  infinity norm). Consequently, the corresponding robust schedule will still provide a ``sufficiently good" approximate solution for the randomly obtained instance. Furthermore, the notion of ``sufficiently good" is also quantitatively described in terms of the properties of the probability distribution over the uncertain delivery times as well as the algorithm's runtime.

As mentioned in \cite{sotskov2010scheduling}, when coping with uncertain environments there are two important phases: the \emph{training phase}, when the machine samples instances from an unknown distribution, and the \emph{practical phase} when the machine must decide on the most suitable solution given an instance of the problem. The decision making in the practical phase is a lot more time-sensitive than in the training phase. It is easy to see that if one has polynomially many parallel processors, it requires only linear run-time to check which individual within the final population contains a suited solution. Of course, our EDA is designed to implement the training phase and we have shown that under the circumstances briefly described above, its runtime is polynomial.

In the future research we plan to extend the methodology presented in the current work to include random release dates and random processing times. Furthermore, we believe that the type of the novel EDA-based methodology designed for the single machine scheduling problem with uncertain delivery times in the current paper can also be adopted to other optimization problems in uncertain environments.


\begin{thebibliography}{10}

\bibitem{auger2011theory}
A.~Auger and B.~Doerr.
\newblock {\em Theory of Randomized Search Heuristics: Foundations and Recent
  Developments}, volume~1.
\newblock World Scientific, 2011.

\bibitem{bianchi2009survey}
L.~Bianchi, M.~Dorigo, L.~M. Gambardella, and W.~J. Gutjahr.
\newblock A survey on metaheuristics for stochastic combinatorial optimization.
\newblock {\em Natural Computing: an international journal}, 8(2):239--287,
  2009.

\bibitem{blom1994problems}
G.~Blom.
\newblock {\em Problems and Snapshots from the World of Probability}.
\newblock Springer, 1994.

\bibitem{ebrahimi2013hybrid}
M.~Ebrahimi, S.~Fatemi~Ghomi, and B.~Karimi.
\newblock Hybrid flow shop scheduling with sequence dependent family setup time
  and uncertain due dates.
\newblock {\em Applied Mathematical Modelling}, 2013.

\bibitem{erdos1961classical}
P.~Erdos and A.~Renyi.
\newblock On a classical problem of probability theory.
\newblock {\em Magyar Tud. Akad. Mat. Kutat{\'o} Int. K{\"o}zl},
  6(1-2):215--220, 1961.

\bibitem{fayad2005fuzzy}
C.~Fayad and S.~Petrovic.
\newblock A fuzzy genetic algorithm for real-world job shop scheduling.
\newblock In {\em Innovations in Applied Artificial Intelligence}, pages
  524--533. Springer, 2005.

\bibitem{feldmann2003single}
M.~Feldmann and D.~Biskup.
\newblock Single-machine scheduling for minimizing earliness and tardiness
  penalties by meta-heuristic approaches.
\newblock {\em Computers \& Industrial Engineering}, 44(2):307--323, 2003.

\bibitem{francca2001memetic}
P.~M. Fran{\c{c}}a, A.~Mendes, and P.~Moscato.
\newblock A memetic algorithm for the total tardiness single machine scheduling
  problem.
\newblock {\em European Journal of Operational Research}, 132(1):224--242,
  2001.

\bibitem{gagne2002comparing}
C.~Gagn{\'e}, W.~Price, and M.~Gravel.
\newblock Comparing an aco algorithm with other heuristics for the single
  machine scheduling problem with sequence-dependent setup times.
\newblock {\em Journal of the Operational Research Society}, pages 895--906,
  2002.

\bibitem{goren2008robustness}
S.~Goren and I.~Sabuncuoglu.
\newblock Robustness and stability measures for scheduling: single-machine
  environment.
\newblock {\em IIE Transactions}, 40(1):66--83, 2008.

\bibitem{gu2010novel}
J.~Gu, M.~Gu, C.~Cao, and X.~Gu.
\newblock A novel competitive co-evolutionary quantum genetic algorithm for
  stochastic job shop scheduling problem.
\newblock {\em Computers \& Operations Research}, 37(5):927--937, 2010.

\bibitem{gu2009novel}
J.~Gu, X.~Gu, and M.~Gu.
\newblock A novel parallel quantum genetic algorithm for stochastic job shop
  scheduling.
\newblock {\em Journal of Mathematical Analysis and Applications},
  355(1):63--81, 2009.

\bibitem{gupta1993minimizing}
M.~C. Gupta, Y.~P. Gupta, and A.~Kumar.
\newblock Minimizing flow time variance in a single machine system using
  genetic algorithms.
\newblock {\em European Journal of Operational Research}, 70(3):289--303, 1993.

\bibitem{hochbaum1996approximation}
D.~S. Hochbaum.
\newblock {\em Approximation algorithms for NP-hard problems}.
\newblock PWS Publishing Co., 1996.

\bibitem{jin2005evolutionary}
Y.~Jin and J.~Branke.
\newblock Evolutionary optimization in uncertain environments-a survey.
\newblock {\em Evolutionary Computation, IEEE Transactions on}, 9(3):303--317,
  2005.

\bibitem{lei2011simplified}
D.~Lei.
\newblock Simplified multi-objective genetic algorithms for stochastic job shop
  scheduling.
\newblock {\em Applied Soft Computing}, 11(8):4991--4996, 2011.

\bibitem{lozano2006towards}
J.~A. Lozano.
\newblock {\em Towards a New Evolutionary Computation: Advances on Estimation
  of Distribution Algorithms}.
\newblock Springer, 2006.

\bibitem{mitavskiy2012polynomial}
B.~Mitavskiy and J.~He.
\newblock A polynomial time approximation scheme for a single machine
  scheduling problem using a hybrid evolutionary algorithm.
\newblock In {\em Evolutionary Computation (CEC), 2012 IEEE Congress on}, pages
  1--8. IEEE, 2012.

\bibitem{peng2004parallel}
J.~Peng and B.~Liu.
\newblock Parallel machine scheduling models with fuzzy processing times.
\newblock {\em Information Sciences}, 166(1):49--66, 2004.

\bibitem{santana2009research}
R.~Santana, P.~Larra{\~n}aga, and J.~A. Lozano.
\newblock Research topics in discrete estimation of distribution algorithms
  based on factorizations.
\newblock {\em Memetic Computing}, 1(1):35--54, 2009.

\bibitem{sels2012hybrid}
V.~Sels and M.~Vanhoucke.
\newblock A hybrid genetic algorithm for the single machine maximum lateness
  problem with release times and family setups.
\newblock {\em Computers \& Operations Research}, 39(10):2346--2358, 2012.

\bibitem{sevaux2003genetic}
M.~Sevaux and S.~Dauz{\`e}re-P{\'e}r{\`e}s.
\newblock Genetic algorithms to minimize the weighted number of late jobs on a
  single machine.
\newblock {\em European journal of operational research}, 151(2):296--306,
  2003.

\bibitem{sotskov2010scheduling}
Y.~Sotskov, N.~Sotskova, T.-C. Lai, and F.~Werner.
\newblock {\em Scheduling Under Uncertainty: Theory and Algorithms}.
\newblock Belorusskaya Nauka, 2010.

\bibitem{tasgetiren2009discrete}
M.~F. Tasgetiren, Q.-K. Pan, and Y.-C. Liang.
\newblock A discrete differential evolution algorithm for the single machine
  total weighted tardiness problem with sequence dependent setup times.
\newblock {\em Computers \& Operations Research}, 36(6):1900--1915, 2009.

\bibitem{wang2009robust}
B.~Wang and S.~He.
\newblock Robust optimization model and algorithm for logistics center location
  and allocation under uncertain environment.
\newblock {\em Journal of Transportation Systems Engineering and Information
  Technology}, 9(2):69--74, 2009.

\bibitem{yoshitomi2002genetic}
Y.~Yoshitomi.
\newblock A genetic algorithm approach to solving stochastic job-shop
  scheduling problems.
\newblock {\em International Transactions in Operational Research},
  9(4):479--495, 2002.

\end{thebibliography}
\end{document}